%% file: tark21.tex
\documentclass[submission,copyright,creativecommons]{eptcs}

\providecommand{\event}{TARK 2021} 

\usepackage[T1]{fontenc}
\usepackage[latin1]{inputenc}
\usepackage[x11names,dvipsnames]{xcolor}
\usepackage{amsmath}
\usepackage{amssymb}
\usepackage{verbatim}
\usepackage{booktabs}
\usepackage{multicol}
\usepackage{xspace}
\usepackage{underscore}
\usepackage{upgreek}
\usepackage[textsize=scriptsize]{todonotes}
\usepackage{lineno}
\usepackage{pxfonts}
\usepackage{algorithm}
\usepackage{algpseudocode}

\usepackage[draft]{fixme}
\usepackage[font=small,skip=0pt]{caption}

\newcommand{\titulo}{Uncertainty-Based Semantics for Multi-Agent Knowing How Logics}

\usepackage{hyperref}

\hypersetup{
  pdftitle          = {\titulo},
  pdfstartview      = {FitH},
  pdfpagelayout     = {OneColumn},
  pdfpagelabels     = true,             % for the reader to display the page number as 'ii (4 of 40)' rather than simply '4 of 40'
  plainpages        = false,            % page anchors with formatted form of the page number (i.e., different anchors for pages 'ii' and '2')
  bookmarksnumbered = {true},
  naturalnames      = {true},
  colorlinks        = {true},           % ligas con color; si es falso pone un recuadro que no aparece en la impresion
  linkcolor         = blue,             % color de las ligas dentro del documento (tabla de contenido, secciones, etc.)
  anchorcolor       = red,              % ?
  urlcolor          = NavyBlue,         % color de las ligas externas
  citecolor         = BrickRed          % color de las referencias
}

\usepackage{newmacros}

\renewcommand{\emptyset}{\varnothing}

\title{Uncertainty-Based Semantics for \\ Multi-Agent Knowing How Logics}

\author{%
  Carlos Areces \qquad\qquad Raul Fervari \qquad\qquad Andr\'es R. Saravia
  \institute{FAMAF, Universidad Nacional de C\'ordoba, \& CONICET, Argentina}
  \email{\{carlos.areces,rfervari\}@unc.edu.ar, andresrsaravia@mi.unc.edu.ar}
  \and
  Fernando R. Vel{\'a}zquez-Quesada\institute{ILLC, Universiteit van Amsterdam, The Netherlands}\email{F.R.VelazquezQuesada@uva.nl}
}

\def\titlerunning{\titulo}
\def\authorrunning{Areces, Fervari, Saravia \& Vel{\'a}zquez-Quesada}

\begin{document}

%\linenumbers

\maketitle

\begin{abstract}
  \input{abstract}
\end{abstract}

%------------------------------------------------------------------------------------------------
\section{Introduction}\label{sec:intro}
\input{intro}

%------------------------------------------------------------------------------------------------
\section{A logic of knowing how}\label{sec:kh-lts}
\input{kh-lts}

%------------------------------------------------------------------------------------------------
\section{Uncertainty-based semantics}\label{sec:kh-ults}
\input{kh-ltsu-new}

%------------------------------------------------------------------------------------------------
\section{Final Remarks}\label{sec:final}
\input{final}

\paragraph*{\bf Acknowledgments.} This work is partially supported by projects ANPCyT-PICT-2017-1130,
Stic-AmSud 20-STIC-03 `DyLo-MPC', Secyt-UNC, GRFT Mincyt-Cba,
and
by the Laboratoire International Associ\'e SINFIN.

\bibliographystyle{eptcs}
\bibliography{references}

%------------------------------------------------------------------------------------------------
%\clearpage
%\appendix

%\input{appendix}

\end{document}

%% file: abstract.tex
We introduce a new semantics for a multi-agent epistemic operator of
\emph{knowing how}, based on an indistinguishability relation between
plans. Our proposal is, arguably, closer to the standard presentation
of \emph{knowing that} modalities in classical epistemic logic.  We
study the relationship between this semantics and
previous approaches, showing that our setting is general enough to
capture them. We also define a sound and
complete axiomatization, and investigate the computational complexity
of its model checking and satisfiability problems. 

%%% Local Variables:
%%% mode: latex
%%% TeX-master: "tark21"
%%% End:

%% file: intro.tex
% main = tark21.tex

Epistemic logic~(EL; \cite{Hintikka:kab,RAK}) is a logical formalism tailored
for reasoning about the knowledge of abstract autonomous entities, commonly
called agents (e.g., a human being, a robot, a vehicle). Most standard
epistemic logics deal with an agent's knowledge about the truth-value
of propositions (the notion of \emph{knowing that}). Thus, they focus
on the study of sentences like \emph{``the agent knows that it is sunny in
  Paris''} or \emph{``the robot knows that it is standing next to a
  wall''}.

At the semantic level, EL formulas are typically interpreted over
relational models~\cite{mlbook,blackburn06}: essentially, labeled
directed graphs. The elements of the domain (called \emph{states} or
\emph{worlds}) represent different possible situations. Each agent has
associated a relation (interpreted as an \emph{epistemic
  indistinguishability} relation), used to represent its
uncertainty: related states are considered indistinguishable for the
agent. %\footnote{We will discuss \emph{indistinguishability} in more detail later.}. 
An agent is said to know that a proposition
$\varphi$ is true at a given state $s$ if and only if $\varphi$ holds
in all states she cannot distinguish from $s$. It is typically assumed
that the indistinguishability relation is an equivalence relation. In
spite of its simplicity, this indistinguishability-based
representation of knowledge has several advantages. First, it also
represents the agent's \emph{high-order} knowledge (knowledge about
her own knowledge and that of other agents). Second, it allows a very natural
representation of actions through which knowledge changes (epistemic updates, see, e.g.,~\cite{DELbook,vanBenthem2011ldii}).

In recent years, other patterns of knowledge besides knowing
  that have been investigated (see the discussion in~\cite{Wang16}).  Some
examples are \emph{knowing whether}~\cite{Hart:1996,FWvD15},
\emph{knowing why}~\cite{artemov2008logic,XuW16} and \emph{knowing the
  value}~\cite{GW16,Baltag16,EijckGW17}.  Motivated by different
scenarios in philosophy and AI, languages for reasoning about
\emph{knowing how} assertions~\cite{fantl2012introduction} are
particularly interesting. Intuitively, an agent knows how to achieve
$\varphi$ given $\psi$ if she has the \emph{ability} to guarantee that
$\varphi$ will be the case whenever she is in a situation in which
$\psi$ holds.

There is a large literature connecting
knowing how with logics of knowledge and action (see,
e.g.,~\cite{Mccarthy69,Moore85,Les00,Hoek00,HerzigT06}). However,
these proposals for representing knowing how have been the
target of criticisms. The main issue is that a simple
combination of standard operators expressing \emph{knowing that} and
\emph{ability} (see, e.g.,~\cite{wiebeetal:2003}) does not seem to lead to a
natural notion of knowing how 
(see~\cite{JamrogaA07,Herzig15} for a discussion).
%One requires a \emph{de re} reading: there should be a `course 
%of action' that, when executed in $\psi$-situations, leads only to 
%$\varphi$-states. 

Taking these considerations into account, \cite{Wang15lori,Wang16,Wang2016}
introduced a framework based on a knowing how binary modality
$\kh(\psi,\varphi)$. At the semantic level, this language is also
interpreted over relational models --- called in this context labeled
transition systems (LTSs). But relations do not represent
indistinguishability anymore; they rather describe the actions an
agent has at her disposal (in some sense, her \emph{abilities}).  An
edge labeled $a$ going from state $w$ to state $u$ indicates that the
agent can execute action $a$ to transform state $w$ into $u$. In the
proposed semantics, $\kh(\psi,\varphi)$ holds if and only if there is
a ``plan'' (a sequence of actions satisfying a constraint called
strong executability) in the LTS that unerringly leads from every
$\psi$-state only to $\varphi$-states.  Other variants of this
knowing how operator follow a similar approach
(see~\cite{Li17,LiWang17,FervariHLW17,Wang19a}).

In these proposals, relations are interpreted as the agent's available
actions (as it is done in, e.g., propositional dynamic logic~\cite{HKT00}); 
and the knowing how of an agent is directly
defined by what these actions can achieve. This is in sharp contrast
with EL, where relational models have two kinds of information:
ontic facts about a given situation (represented by the current state in the model), and the particular perspective that agents have
(represented by the possible states available in the model, and their respective indistinguishability relation between them).\footnote{Notice that in a multi-agent scenario, all agents share the same ontic information, and differ on their epistemic interpretation of it. We will come back to this later.}
If one would like to mirror the situation in EL, it seems natural that knowing how should be defined in terms of some kind of indistinguishability over the
 information provided by an LTS.  Such an extended model would be
able to capture both the
abilities of an agent as given by her available actions, together with the (in)abilities that arise when considering two different actions/plans/executions
indistinguishable.

This paper introduces a new semantics for $\kh_i(\psi,\varphi)$, a
multi-agent version of the knowing how modality. The crucial idea is
the inclusion of a notion of epistemic indistinguishability over
plans, in the spirit of the \emph{strategy indistinguishability} of, e.g.,~\cite{JamrogaH04,Belardinelli14}. We interpret formulas over an \emph{uncertainty-based LTS
  (\ults)} which is an LTS equipped with an indistinguishability
relation over plans. An agent may have different alternatives at her disposal to try to achieve a goal, all ``as good as any other'' (and
in that sense indistinguishable) as far as she can tell.
%Consider
%the following example. An agent knows how to get drunk: by drinking
%two bottles of wine, or by drinking three glasses of whiskey. The two
%plans are different, but lead to the same goal: drunkenness. An agent
%might consider them indistinguishable, and together they constitute a
%\emph{strategy} (from suitable preconditions of availability of wine
%or whiskey to the goal state of being drunk).  
In this way, \ults{s} aims to reintroduce the notion of epistemic
indistinguishability, now at the level of plans. 
Moreover, the use of \ults{s} leads to a
natural definition of operators that represent dynamic aspects of
knowing how (e.g., the concept of \emph{learning how}
can be modeled by eliminating uncertainty between plans).

\ssparagraph{Our contributions.} They can be summarized as follows: \begin{inlineenum} \item We introduce a new semantics for $\kh_i(\psi,\varphi)$ (for $i$ an agent) that reintroduces the notion of epistemic indistinguishability from classical EL. \item We show that the logic obtained is weaker (and this is an advantage, as we will discuss) than the logic from~\cite{Wang15lori,Wang16,Wang2016}. Still, the new semantics is general enough to capture previous proposals by imposing adequate conditions on the class of models. %
%
%\item We introduce suitable bisimulations for the new semantics and prove that, over arbitrary models, the logic obtained is weaker than the logic from~\cite{Wang15lori,Wang16,Wang2016}.
%
\item We present a sound and complete axiomatization for the logic over the class of all \ults{s}.
\item We prove that the satisfiability problem for the new logic is \NP-complete, whereas model checking is in \Poly.
%\item We exemplify the dynamic capabilities of the new semantics introducing modalities that update the \emph{knowing how} of an
 % agent.
\end{inlineenum}

  \ssparagraph{Outline.} \Cref{sec:kh-lts} recalls the framework
  of~\cite{Wang15lori,Wang16,Wang2016}, including its axiom
  system. \Cref{sec:kh-ults} introduces \emph{uncertainty-based LTS},
  indicating how it can be used for interpreting a multi-agent version
  of the knowing how language, and providing an axiom
  system in \Cref{sec:axiom}. \Cref{subsec:comparing} studies the correspondence between
  our semantics and the ones in the previous proposals. \Cref{subsec:complexity} studies the
  computational complexity of model checking and the satisfiability problem for our
  logic. \Cref{sec:final} provides conclusions and future
  lines of research.

%%% Local Variables:
%%% mode: latex
%%% TeX-master: "tark21"
%%% End:

%% file: kh-lts.tex
%!TEX root = tark21.tex
This section recalls the basic \emph{knowing how} framework from \cite{Wang15lori,Wang16,Wang2016}.

\ssparagraph{Syntax and semantics.} Throughout the text, let $\PROP$ be a countable non-empty set of propositional symbols.

\begin{definition}\label{def:khsyntax}
  Formulas of the language \KHlogic are given by the following grammar:
  \begin{nscenter}
    $\varphi ::= p \mid \neg\varphi \mid \varphi\vee\varphi \mid \kh(\varphi,\varphi)$,
  \end{nscenter}
  with $p \in \PROP$. Other Boolean connectives are defined as usual. Formulas of the form $\kh(\psi,\varphi)$ are read as \emph{``when $\psi$ holds, the agent knows how to make $\varphi$ true''}.
\end{definition}

In \cite{Wang15lori,Wang16,Wang2016} (and variations like \cite{LiWang17,Li17}), formulas of $\KHlogic$ are interpreted over \emph{labeled transition systems}: relational models in which the relations describe the state-transitions available to the agent. Throughout the text, let $\ACT$ be a denumerable set of (basic) action names. %First, some syntactic notions and then the structure.

\begin{definition}[Actions and plans]
  Let $\ACT^*$ be the set of finite sequences over $\ACT$. Elements of
  $\ACT^*$ are called \emph{plans}, with $\epsilon$ being the
  \emph{empty plan}. %Let $\ACT^+ := \ACT^* \setminus  \cset{\epsilon}$. 
  Given $\sigma \in \ACT^*$, let $\card{\sigma}$ be
  the length of $\sigma$ ($\card{\epsilon} := 0$). For
  $0 \le k \le \card{\sigma}$, the plan $\sigma_k$ is $\sigma$'s initial
  segment up to (and including) the $k$th position (with
  $\sigma_0 := \epsilon$). For $0 < k \le \card{\sigma}$, the action
  $\sigma[k]$ is the one in $\sigma$'s $k$th position.
\end{definition}

\begin{definition}[Labeled transition systems]\label{def:abmap}
  A \emph{labeled transition system (\lts)} for $\PROP$ and $\ACT$ is
  a tuple $\modlts = \tup{\W,\R,\V}$ where $\W$ is a non-empty set of
  states (also denoted by $\D{\modlts}$),
  $\R = \csetsc{\R_a \subseteq \W \times \W}{a \in \ACT}$ is a
  collection of binary relations on $\W$, and $\V:\W \to 2^\PROP$ is a
  labelling function. Given an \lts $\modlts$ and $w \in \D{\modlts}$,
  the pair $(\modlts,w)$ is a \emph{pointed \lts} (parentheses are
  usually dropped).
\end{definition}

An \lts describes the \emph{abilities} of the agent; thus, sometimes (e.g., \cite{Wang15lori,Wang16,Wang2016}) it is also called an \emph{ability map}. Here are some useful definitions.

\begin{definition}
  Let $\csetsc{\R_a \subseteq \W \times \W}{a \in \ACT}$ be a collection of binary relations. Define $\R_\epsilon := \csetsc{(w,w)}{w \in \W}$ and, for $\epsilon\neq\sigma \in \ACT^*$ and $a \in \ACT$, $\R_{\sigma{a}} := \csetc{(w,u)}{\W \times \W}{ \exists v \in \W \text{ s.t. } (w,v) \in \R_{\sigma} \text{ and } (v,u) \in \R_{a} }$. Take a plan $\sigma \in \ACT^*$: for $u \in \W$ define $\R_{\sigma}(u) := \csetsc{v \in \W}{(u,v) \in \R_{\sigma}}$, and for $U \subseteq \W$ define $\R_{\sigma}(U) := \bigcup_{u \in U} \R_{\sigma}(u)$.
\end{definition}

Intuitively, \cite{Wang15lori,Wang16,Wang2016} defines that an agent
knows how to achieve $\varphi$ given $\psi$ when she has an
appropriate plan that allows her to go from any situation in which
$\psi$ holds to only states in which $\varphi$
holds. A crucial part is, then, what ``appropriate'' is taken to be.

%Still, one can impose different restrictions on what qualifies as an adequate plan.

\begin{definition}[Strong executability]\label{def:plans-exec}
  Let $\csetsc{\R_a \subseteq \W \times \W}{a \in \ACT}$ be a collection of binary relations. A plan $\sigma \in \ACT^*$ is 
%\emph{weakly executable} at a given $u \in \W$ if and only if $\R_\sigma(u) \neq \emptyset$. The set $\wkexec(\sigma)$ contains the states in $\W$ where $\sigma$ is weakly executable.
  \emph{strongly executable} (SE) at $u \in \W$ if and only if $v \in \R_{\sigma_k}(u)$ implies $\R_{\sigma[k+1]}(v) \neq \emptyset$ for every $k \in \intint{0}{\card{\sigma}-1}$. We define the set $\stexec(\sigma):= \cset{w\in\W \mid \sigma \mbox{ is SE at }w}$. % contains the states in $\W$ where $\sigma$ is strongly executable.
\end{definition}

Thus, %while weak executability asks only for \emph{some} partial execution of the plan to be completed,
strong executability asks for \emph{every} partial execution of the plan (which might be $\epsilon$) to be completed. %Note: the notion of weak executability is not used in the original \emph{knowing how} proposals. It is introduced her for the purposes of comparisons and as a basis for other executability notions. It is not hard to see that $\stexec(a) = \wkexec(a)$ for every $a \in \ACT$, and that $\stexec(\sigma) \subseteq \wkexec(\sigma)$ for every $\sigma \in \ACT^*$.
With this notion, formulas in $\KHlogic$ are interpreted over an LTS as follows (the semantic clause for the $\kh$ modality is equivalent to the one found in the original papers).

\begin{definition}[\KHlogic over \ltss]\label{def:sem-abmap}
  The relation $\models$ between a pointed \lts $\modlts, w$ (with $\modlts=\tup{\W,\R,\V}$ an \lts over \ACT and \PROP) and formulas in \KHlogic (over \PROP) is defined inductively as follows:
  % The cases for atoms and Boolean operators are as usual. For \emph{knowing how} formulas,
 \begin{nscenter}
   \begin{tabular}{@{}lcl@{}}
     $\modlts,w \models p$ & \iffdef & $w\in\V(p)$, \\
     $\modlts,w \models \neg\varphi$ & \iffdef & $\modlts,w \not\models\varphi$, \\ 
     $\modlts,w \models \varphi\vee\psi$ & \iffdef & $\modlts,w \models \varphi \,\mbox{ or }\, \modlts,w \models\psi$, \\
     $\modlts,w \models \kh(\psi,\varphi)$ & \iffdef & \begin{minipage}[t]{0.725\textwidth}
                                                         $\exists \sigma \in \ACT^*$ such that \begin{inline-cond-kh} \item $\truthset{\modlts}{\psi} \subseteq \stexec(\sigma)$ and \item $\R_\sigma(\truthset{\modlts}{\psi}) \subseteq \truthset{\modlts}{\varphi}$,
                                                         \end{inline-cond-kh}
                                                       \end{minipage}
   \end{tabular}
 \end{nscenter}
  % \begin{nscenter}
  %     $\modlts,w \models p ~\iffdef~ p\in\V(w)$, \quad\quad
  %     $\modlts,w \models \neg\varphi  ~\iffdef~  \modlts,w \not\models\varphi$, \quad\quad 
  %     $\modlts,w \models \varphi\vee\psi  ~\iffdef~  \modlts,w \models \varphi \mbox{ or } \modlts,w \models\psi$, \\
  %     $\modlts,w \models \kh(\psi,\varphi)$  ~\iffdef~  \begin{minipage}[t]{0.725\textwidth}
  %                                                         $\exists \sigma \in \ACT^*$ such that \begin{inline-cond-kh} \item $\truthset{\modlts}{\psi} \subseteq \stexec(\sigma)$ and \item $\R_\sigma(\truthset{\modlts}{\psi}) \subseteq \truthset{\modlts}{\varphi}$,
  %                                                         \end{inline-cond-kh}
  %                                                       \end{minipage}
  % \end{nscenter}
  with $\truthset{\modlts}{\varphi} := \csetc{w}{\W}{\modlts,w \models \varphi}$ (the elements of $\truthset{\modlts}{\varphi}$ are sometimes called $\varphi$-states). %A formula $\varphi$ is satisfiable iff $\lts,w\models\varphi$ for some $\lts$ and $w\in \D{\lts}$, and it is valid (notation $\models\varphi$) iff $\lts,w\models\varphi$ for all $\lts$ and all $w \in \D{\lts}$.  A formula $\varphi$ holds (globally) in an LTS $\lts$ ($\lts \models \varphi$) iff $\lts,w \models \varphi$ for all $w \in \D{\lts}$.  For $\Psi \subseteq \KHlogic$, $\lts,w\models\Psi$ iff $\lts,w\models\psi$, for all $\psi\in\Psi$;   and $\Psi \models\varphi$ iff for all $\lts$, $w \in \D{\lts}$, $\lts,w\models\Psi$ implies $\lts,w\models\varphi$.
\end{definition}

$\kh(\psi,\varphi)$ holds when there is a plan $\sigma$ such that,
when it is executed at any $\psi$-state, it will always complete every partial
execution (condition \ITMKHi), ending unerringly in states satisfying
$\varphi$ (condition \ITMKHii).  Notice that $\kh$ acts \emph{globally},
i.e.,  $\truthset{\modlts}{\kh(\psi,\varphi)}$ is either
$\D{\modlts}$ or  $\emptyset$.

\msparagraph{Axiomatization.} The universal modality
\cite{GorankoP92}, interpreted as truth in every state of the model,
is definable in $\KHlogic$ as $\A\varphi := \kh(\neg\varphi,\bot)$. This is justified by the following proposition, whose proof relies on the fact that $\ACT^*$ is never empty (it always contains $\epsilon$).

\begin{proposition}[\cite{Wang15lori}]\label{prop:lts:universal}
  Let $\modlts, w$ be a pointed \lts. Then,
  $\modlts, w \models \kh(\neg\varphi,\bot)$ iff $\truthset{\modlts}{\varphi} = \D{\modlts}$. 
%   \begin{center}
%     $\modlts, w \models \kh(\neg\varphi,\bot)$ \qquad iff \qquad $\truthset{\modlts}{\varphi} = \D{\modlts}$. 
%   \end{center} 
\end{proposition}

%Define $\A\varphi := \kh(\neg\varphi,\bot)$. 

\begin{table}[h!]
  \begin{cfootnotesizetabular}{ll@{\;}ll@{\;}l}
    \toprule
    \underline{Block $\axset$:}        & \axm{TAUT}   & $\vdash \varphi$ \; for $\varphi$ a propositional tautology &
                                         \axm{MP}     & From $\vdash \varphi$ and $\vdash \varphi \limp \psi$ infer $\vdash \psi$ \\
                                       & \axm{DISTA}  & $\vdash \A(\varphi\ra\psi) \ra (\A\varphi \ra \A\psi)$ &
                                         \axm{NECA}   & From $\vdash \varphi$ infer $\vdash \A\varphi$ \\
                                       & \axm{TA}     & $\vdash \A\varphi \ra \varphi$ \\
                                       & \axm{4KhA}   & $\vdash \kh(\psi,\varphi) \ra \A\kh(\psi,\varphi)$ &
                                         \axm{5KhA}   & $\vdash \neg\kh(\psi,\varphi) \ra \A\neg\kh(\psi,\varphi)$ \\
    \midrule
    \underline{Block $\axset_{\lts}$:} & \axm{EMP}    & $\vdash \A(\psi \ra \varphi) \ra \kh(\psi,\varphi)$ &
                                         \axm{COMPKh} & $\vdash (\kh(\psi,\varphi)\wedge\kh(\varphi,\chi))\ra\kh(\psi,\chi)$ \\
    \bottomrule
  \end{cfootnotesizetabular}
  \caption{Axiom system \KHaxiom, for \KHlogic w.r.t. \ltss.}\label{tab:khaxiom}
\end{table}

The axiom system \KHaxiom in~\Cref{tab:khaxiom} shows that $\A$ and $\kh$ are strongly interconnected. 

\begin{theorem}[\cite{Wang15lori}]\label{teo:kh-sound-complete-lts}
   \KHaxiom is sound and strongly complete for \KHlogic w.r.t. the class of all \ltss.
\end{theorem}

Some axioms deserve comment.  If $\A$ is taken as primitive, and
$\A\varphi$ is interpreted as $\varphi$ is true at every state in an
\lts, then $\axm{EMP}$ states that if $\psi \to \varphi$ is a globally true implication, then given $\psi$ the agent has the ability to make $\varphi$ true. In simpler words, global ontic
information turns into knowledge.  One could argue that, more
realistically, there are global truths in a model that are still
beyond the abilities of the agent. The case of $\axm{COMPKh}$ is
similar (as it also implies a certain level of omniscience) but
perhaps less controvertial.  It might well be that an agent knows how
to make $\varphi$ true given $\psi$, and how to make $\chi$ true given
$\varphi$, but still have not workout how to put the two together to
ensure that $\chi$ given $\psi$.  As we will see in the next section,
both these axioms can be correlated with strong assumptions on the
uncertainty relation between plans that an agent might have.

%%% Local Variables:
%%% mode: latex
%%% TeX-master: "tark21"
%%% End:

%% file: kh-ltsu-new.tex
%!TEX root = tark21.tex
The \lts-based semantics provides a possible representation of an
agent's abilities: the agent knows how to achieve $\varphi$ given
$\psi$ if and only if there is a plan that, when run at any
$\psi$-state, will always complete every partial execution, ending
unerringly in states satisfying $\varphi$.  One could argue that
this representation involves a certain level of idealization. 

Consider an agent \emph{lacking} a certain ability. In the \lts-based
semantics, this can only happen when the environment does not provide the required
(sequence of) action(s). But one can think of scenarios where an adequate plan
exists, and yet the agent lacks the ability for a different reason. Indeed, she might
\emph{fail to distinguish} the adequate plan from a non-adequate one, in the
sense of not being able to tell that, in general, those plans produce
a different outcome. Consider, for example, an agent baking a
cake. She might have the ability to do the
\href{https://www.perfectlypastry.com/the-importance-of-the-mixing-method/}{nine
  different mixing methods} (beating, blending, creaming, cutting,
folding, kneading, sifting, stirring, whipping), and she might even
recognize them as different actions. However, she might not be able to
perfectly distinguish one from the others in the sense of not recognizing that
sometimes they produce different results. In such cases, one would say that the
agent does not know how to make certain cake: sometimes she gets good
outcomes (when she uses the adequate mixing method) and sometimes she
does not.

Indistinguishability among \emph{basic} actions can account for the
example above (with each mixing method a basic action). Still, one can also
think of situations in which a more general indistinguishability \emph{among plans} is involved.  Consider the baking agent
again. It is reasonable to assume that she can tell the difference
between ``adding milk'' and ``adding flour'', but perhaps she does not
realize the effect that \emph{the order} for mixing ingredients might
have in the final result. Here, the issue is not that she cannot
distinguish between basic actions; rather, two plans are
indistinguishable because the order of their actions is being
  considered irrelevant. For a last possibility, the agent might
not know that, while opening the oven once to check whether the baking
goods are done is reasonable, this must not be done in excess. In this
case, the problem is not being able to tell the difference between
the effect of executing an action once and executing it multiple times.
Thus, even plans of \emph{different length}
might be considered indistinguishable.
  
The previous examples suggest that one can devise a more general
representation of an agent's abilities. This representation involves
taking into account not only
the plans she has available (the \lts structure), but also her skills
for telling two different plans apart (a form of
\emph{indistinguishability among plans}).
%Moreover: although
%indistinguishability among plans aims to deal with additional reasons
%why an agent might lack certain ability, it also allows a more
%flexible representation of the abilities the agent does have. The
%baking agent might sometimes use a red bowl and sometimes use a green
%one, not being able to tell the difference between the two
%actions. Thus, she has actually two plans that she considers
%indistinguishable, and yet this still works because both plans produce
%always the same outcome.  The semantic structure below puts these
%ideas to work.
As we will see, the use of an indistinguishability relation among
plans will also let us define a natural model for a multi-agent
scenario. In this setting, agents share the same set of \emph{affordances} (provided
by the actual environment), but still have different \emph{abilities} depending
on which of these affordances they have available, and how well they can tell these affordances appart.

\begin{definition}[Uncertainty-based \lts]\label{def:ults}
  Let \AGT be a finite non-empty set of agents.  A \emph{multi-agent
    uncertainty-based \lts (\ults)} for $\PROP$, $\ACT$ and $\AGT$ is
  a tuple $\modults = \tup{\W,\R,\sim,\V}$ where $\tup{\W, \R, \V}$ is
  an \lts and $\sim$ assigns, to each agent $i \in \AGT$, an
  equivalence \emph{indistinguishability} relation over a non-empty
  set of plans $\DS{i} \subseteq \ACT^*$. Given an \ults $\modults$
  and $w \in \D{\modults}$, the pair $(\modults,w)$ (parenthesis
  usually dropped) is called a \emph{pointed \ults}.
\end{definition}

Intuitively, $\DS{i}$ is the set of plans that agent $i$ has at her
disposal. Similarly as in classical epistemic logic,
${\sim_i} \subseteq \DS{i} \times \DS{i}$ describes agent $i$'s
indistinguishability. But this time, this relation is not defined over possible 
states of affairs, but rather over her available plans.

\begin{remark}\label{rem:corresp}
  The following change in notation will simplify some definitions
  later on, and will make the comparison with the \lts-based semantics
  clearer.  Let $\tup{\W, \R, \sim, \V}$ be an \ults and take
  $i \in \AGT$; for a plan $\sigma \in \DS{i}$, let $[\sigma]_i$ be
  its equivalence class in $\sim_i$ (i.e.,
  $[\sigma]_i := \csetc{\sigma'}{\DS{i}}{ \sigma \sim_i
    \sigma'}$).

  There is a one-to-one correspondence between $\sim_i$ and its induced
  set of equivalence classes
  $\S_i := \csetsc{[\sigma]_i}{\sigma \in \DS{i}}$. Hence, from now on
  an \ults will be presented as a tuple $\tup{\W, \R, \S, \V}$, with
  $\S = \csetsc{\S_i}{i \in \AGT}$. Notice the following
  properties: \begin{inlineenum} \item
    $\strategy_1\neq\strategy_2 \in \S_i$ implies
    $\strategy_1 \cap \strategy_2 = \emptyset$, \item
    $\DS{i} = \bigcup_{\strategy \in \S_i} \strategy$ and \item
    $\emptyset \notin \S_i$\end{inlineenum}.
\end{remark}

Given her uncertainty over $\ACT^*$, the abilities of an agent $i$
depend not on what a single plan can achieve, but rather on what a set
of them can guarantee.

\begin{definition}
  For $\strategy \subseteq \ACT^*$ and $U \cup \{u \} \subseteq \W$
  define $\R_\strategy := \bigcup_{\sigma \in \strategy} \R_{\sigma}$,
  $\R_{\strategy}(u) := \bigcup_{\sigma \in \strategy} \R_\sigma(u)$,
  and $\R_{\strategy}(U) := \bigcup_{u \in U} \R_{\strategy}(u)$.
\end{definition}

We can now define strong executability for sets of plans.

\begin{definition}[Strong executability]\label{def:strat-exec}
  A set of plans $\strategy \subseteq \ACT^*$ is \emph{strongly
    executable} at $u \in \W$ if and only if \emph{every} plan
  $\sigma \in \strategy$ is \emph{strongly executable} at $u$.
  Hence, $\stexec(\strategy) = \bigcap_{\sigma \in \strategy}
  \stexec(\sigma)$ is the set of the states in $\W$ where $\strategy$ is strongly
  executable.
\end{definition}

%There are other possibilities for defining adequate sets of plans. One might require for the \emph{existence} of a plan in $\strategy$ that is \emph{strongly} executable (`$\bigcup \stexec$'), or for \emph{all} plans in $\strategy$ to be \emph{weakly} executable (`$\bigcap \wkexec$'). This paper will focus on the two extreme cases defined above.

\begin{definition}[$\khi$ over \ultss]\label{def:sem-esm}
  The satisfiability relation $\models$ between a pointed \ults $\modults, w$
  (with $\modults = \tup{\W, \R, \S, \V}$ an \ults over \ACT, \PROP and
  \AGT) and formulas in the multi-agent version of $\KHlogic$ (denoted
  by $\KHilogic$, and obtained by replacing $\kh$ with $\khi$, $i\in\AGT$) is
  defined inductively. The atomic and Boolean cases are as before. For
  \emph{knowing how} formulas,
  \begin{nscenter}
    \begin{tabular}{@{}lcl@{}}
      $\modults,w \models \khi(\psi,\varphi)$ & \iffdef & \begin{minipage}[t]{0.725\textwidth}
                                                            $\exists \strategy \in \S_i$ such that \begin{inline-cond-kh} \item $\truthset{\modults}{\psi} \subseteq \stexec(\strategy)$ and \item $\R_\strategy(\truthset{\modults}{\psi}) \subseteq \truthset{\modults}{\varphi}$,
                                                            \end{inline-cond-kh}
                                                          \end{minipage}
    \end{tabular}
  \end{nscenter}
  with $\truthset{\modults}{\varphi} := \csetc{w}{\W}{\modults,w \models \varphi}$.
\end{definition}

It is worth comparing~\Cref{def:sem-abmap} and~\Cref{def:sem-esm}.
As before, $\khi(\psi,\varphi)$ acts
\emph{globally}. Moreover, we now require \emph{for
  agent $i$} to have a \emph{set of plans} satisfying strong
executability in every $\psi$-state (condition \ITMKHi). Still, the
set of plans should work as the single plan did before: when executed at
$\psi$-states, it should end unerringly in states satisfying $\varphi$
(condition \ITMKHii).

The rest of the section is devoted to explore the properties of the logic with our
new semantics. Moreover, we compare it to the well-known framework from~\cite{Wang15lori,Wang16,Wang2016}.

%\msparagraph{Axiomatization.} 
\subsection{Axiomatization} \label{sec:axiom}

We start by establishing that the universal 
modality is again definable within $\KHilogic$ over
\ults (it is crucial that $\S_i \neq \emptyset$ and
$\emptyset \not\in \S_i$, see \Cref{rem:corresp}).  

\begin{proposition}\label{pro:ults:universal}
  Given $\modults, w$  a pointed \ults, then ($\exists i \in \AGT$,
   $\modults, w \models \khi(\neg\varphi,\bot)$) iff $\truthset{\modults}{\varphi} = \D{\modults}$. 
  %\begin{center}
   % $\modults, w \models \khi(\neg\varphi,\bot)$ \qquad iff \qquad $\truthset{\modults}{\varphi} = \D{\modults}$. 
  %\end{center}
\end{proposition}

Hence, by taking $\A \varphi := \bigvee_{i \in \AGT} \khi(\neg\varphi,\bot)$ (recall
that $\AGT$ is non-empty and finite) and $\E\varphi:=\neg\A\neg\varphi$, it turns out that formulas in $\axset$
(first part of \Cref{tab:khaxiom}) are still valid, generalizing $\kh$ to
$\khi$.  As discussed in the next section, some valid formulas in $\axset_{\lts}$ can be
falsified over \ults{s}. But the weaker theorems of \KHaxiom shown
in~\Cref{tab:khiaxiom} (see \Cref{pro:subsystem}) are still valid, and can be used to define a
complete axiomatic system.

\begin{table}[h!]
  \begin{cfootnotesizetabular}{lll@{\qquad}ll}
    \toprule
    \underline{Block $\axset_{\ults}$:} & \axm{KhE} & $\vdash \left(\E\psi \land \khi(\psi,\varphi)\right) \limp \E\varphi$ &
                                          \axm{KhA} & $\vdash \left(\A(\chi \limp \psi) \land \khi(\psi,\varphi) \land \A(\varphi \limp \theta)\right) \limp \khi(\chi, \theta)$ \\
    \bottomrule
  \end{cfootnotesizetabular}
  \caption{Axioms $\axset_{\ults}$, for \KHilogic w.r.t. \ultss.}\label{tab:khiaxiom}
\end{table}

\axm{KhA} can be subjected to some of the criticism that apply to
\axm{EMP} and \axm{COMPKh} but, in our opinion, to a lesser extent. It
 implies certain level of idealization, as it entails that the
knowing how of an agent is, in a sense, closed under global
entailment. \axm{KhE} on the other hand, seems  plausible:
if $\khi(\psi,\varphi)$ is not trivial (given that $\E\psi$ holds),
then $\E\varphi$ should be assured.

\input{section-completeness-body}

\begin{theorem}\label{teo:khi-sound-complete-ults}
  The axiom system \KHiaxiom := $\axset$  + $\axset_{\ults}$ (\Cref{tab:khiaxiom}) is sound and strongly complete for \KHilogic w.r.t. the class of all \ultss.
\begin{proof}
Take any \KHiaxiom-consistent set of formulas $\Gamma' \subseteq \KHilogic$. Since \KHilogic is enumerable, $\Gamma'$ can be extended into a maximally \KHiaxiom-consistent set $\Gamma \supseteq \Gamma'$ by a standard Lindenbaum's construction (see, e.g., \cite[Chapter 4]{mlbook}). By \Cref{tlm:cm-ults-lkhi}, $\Gamma'$ is satisfiable in $\modults^\Gamma$ at $\Gamma$. The fact that $\modults^\Gamma$ is in \ults (\Cref{pro:cm-ults-lkhi}) completes the proof. \closeproofapx
\end{proof}
\end{theorem}

%---------------------------------------
\subsection{Comparing \texorpdfstring{\lts}{LTS} semantics and \texorpdfstring{\ults}{LTSU} semantics} \label{subsec:comparing}
%---------------------------------------

The provided axiom system can be used to compare the notion of \emph{knowing how} under \ltss with that under \ultss. Here is a first observation.

\begin{proposition}\label{pro:subsystem}
  Axioms \axm{KhE} and \axm{KhA} are \KHaxiom-derivable (thus, $\KHiaxiom$ is a subsystem of $\KHaxiom$).
%  \begin{proof}
%    Axiom \axm{KhE} can be rewritten as $\left(\khi(\psi,\varphi) \land \A\lnot\varphi\right) \limp \A\lnot\psi$, which by unfolding $\A$ is an instance of \axm{COMPKh} in \KHaxiom. For \axm{KhA}, use \axm{EMP} and then \axm{COMPKh} (\cite[Proposition 2]{Wang15lori}).   
%    \end{proof}
\end{proposition}

Hence, the \emph{knowing how} operator under \lts is at least as strong as its \ults-based counterpart: every formula valid under \ultss is also valid under \ltss. The following fact shows that the converse is not the case.

\begin{proposition}\label{fact:axiom-fail}
  Within \ults, axioms \axm{EMP} and \axm{COMPKh} are not valid.
  \begin{proof}
    Consider the \ults $\model$ shown below, with the collection of sets of available plans for agent $i$ (i.e., the set $\S_i$) depicted on the right. Recall that $\khi$ acts globally.
    \begin{nscenter}
      \begin{tikzpicture}[->]
        \node [state, label = {[label-state]left:$w$}] (w1) {$p$};      
        \node [state, right = of w1] (w2) {$q$};
        \node [state, right = of w2] (w3) {$r$};
        \node [state, above = of w3] (w4) {};
        
        \path (w1) edge node [label-edge, below] {$a$} (w2)
                   edge node [label-edge, above] {$c$} (w4)
              (w2) edge node [label-edge, below] {$b$} (w3);
      \end{tikzpicture}
      \hspace{1.5cm}
      \begin{picture}(60,0)
        \put(-20,20){
          $\S_i = \left\{
            \begin{array}{c}
              \{a\},\ \{b\}\\
              \{ab, c\}
            \end{array}
            \right\}$}
        \end{picture}
    \end{nscenter}
    With respect to \axm{EMP}, notice that $\A(p\ra p)$ holds; % (recall: $\A$ is interpreted as truth in every state); 
    yet, $\khi(p,p)$ fails since there is no $\strategy \in \S_i$ leading from $p$-states to $p$-states. More generally, \axm{EMP} is valid over \ltss because the empty plan $\epsilon$, strongly executable everywhere, is always available. However, in a $\ults$, the plan $\epsilon$ might not be available to the agent (i.e., $\epsilon \notin \DS{i}$); and even if $\epsilon$ is avaibable, it might be indistinguishable from other plans with different behaviour.
      % \Cref{tab:khaxiom}.  In all pointed LTS $\modlts,w$, if
      % $\modlts,w\models\A(p\ra q)$ then $\modlts,w\models\kh(p,q)$. As shown in
      % \cite{Wang15lori}, the witness SE plan for satisfying $\kh(p,q)$ when
      % $\modlts,w\models\A(p\ra q)$ is the empty plan $\epsilon$.  However, in an \ults $\model=\tup{\W,\R,\S,\V}$, 
      % the plan $\epsilon$ could 
      % belong to a strategy $\strategy\in\S$ containing other plans leading to $\neg q$-states, even if $\model \models \A(p\ra q)$. In fact,
      % $\epsilon$ might not even be in $\DS{\S}$. Thus, it is easy to
      % find counter examples to {\sf EMP} over {\ults}s.
      
  With respect to \axm{COMPKh}, notice that $\khi(p,q)$ and $\khi(q,r)$ hold, witness $\cset{a}$ and $\cset{b}$, resp. However, there is no $\strategy\in\S_i$ containing only plans that, when starting on $p$-states, lead only to $r$-states. This is due to the fact that, although $ab$ acts as needed, it cannot be distinguished from $c$, which behaves differently. Thus, $\khi(p,r)$ fails. More generally, \axm{COMPKh} is valid over \lts because the sequential composition of the plans that make true the conjuncts in the antedecent is a witness that makes true the consequent. However, in an \ults, this composition might be unavailable or else indistinguishable from other plans.
%  as we just showed, this is not the case in an \ults.
%  in an \ults, the existence of two sets of plans $\strategy_1$ and $\strategy_2$ does not guarantee the existence of a third one containing the sequential composition of the plans in $\strategy_1$ with the plans in $\strategy_2$. And even if such a third set of plans exists, it might contain additional plans that have a non-adequate behaviour.
  \end{proof}
\end{proposition}

From these two observations it follows that $\kh$ under \ultss is strictly weaker than $\kh$ under \ltss: adding uncertainty about the effect of actions does change the logic. However, the \ults framework is general enough to capture the \lts semantics. To establish the connection, let us work in a single-agent setting (i.e., with a single modality $\kh$ and no subindexes for $\DS{i}$ and $\S_i$).

{\smallskip}

Given the discussion in \Cref{fact:axiom-fail}, it should be clear that there is an obvious class of \ultss in which \axm{EMP} and \axm{COMPKh} are valid. This is the class of \ultss in which the agent has every plan available %(i.e., $\DS{} = \ACT^*$) 
and can distinguish between any two of them (i.e., $\S = \csetsc{\cset{\sigma}}{\sigma \in \ACT^*}$). This is because, in such models, $\epsilon$ is available and distinguishable from other plans (for \axm{EMP}) and from $\cset{\sigma_1} \in \S$ and $\cset{\sigma_2} \in \S$ it follows that $\cset{\sigma_1\sigma_2} \in \S$ (for \axm{COMPKh}).   
%In other words, the abilities of an agent in \lts are exactly as the abilities of an agent in \ults that can use every plan and has no uncertainty on them. 
Clearly, other, more general, classes can be defined, but the one introduced here serves as an example.

\begin{proposition}
  Let $\modlts=\tup{\W,\R,\V}$ be an LTS over $\ACT$, define
  $\model_\modlts = \tup{\W,\R, \S, \V}$, where
  $\S=\csetsc{\cset{\sigma}}{\sigma\in\ACT^*}$.  Let
  $\cults:=\csetsc{\model_\modlts}{\modlts \mbox { is an LTS}}$.
  Given $\model=\tup{\W,\R,\S,\V}$ an \esm in $\cults$, define
  $\modlts_\model=\tup{\W,\R,\V}$.  Then, for every
  $\varphi \in \KHlogic$,
  $\truthset{\modlts}{\varphi} = \truthset{\model_\modlts}{\varphi}$
  and
  $\truthset{\model}{\varphi} = \truthset{\modlts_\model}{\varphi}$.
\end{proposition}

%Thus, any \lts can be transformed into a pointwise equivalent \ults by taking its strongly executable plans (the ones that define her abilities, which include $\epsilon$) and turn them into basic actions that can be distinguished from any other. 

%\begin{textonuevo}
%Thus, any \ults \emph{in $\cults$} can be transformed into a pointwise equivalent \lts by collapsing each set of plans available to the agent into a single basic action.
%\end{textonuevo}
Since we have a class of \ultss in correspondence with the class of
all \ltss, we get a direct completeness result:

\begin{theorem}\label{teo:kh-sound-complete-cults}
  The axiom system \KHaxiom (\Cref{tab:khaxiom}) is sound and strongly complete for \KHlogic w.r.t. the class $\cults$. % of \ults{s} $\cults$. 
\end{theorem}

\subsection{Complexity}\label{sec:complexity}
\label{subsec:complexity}
\input{complexity}

%%% Local Variables:
%%% mode: latex
%%% TeX-master: "tark21"
%%% End:

%% file: section-completeness-body.tex
%!TEX root = tark21.tex

Let us define the system \KHiaxiom := $\axset$  + $\axset_{\ults}$ (\Cref{tab:khiaxiom}).
We will show that the system is sound and strongly complete over \ults{s}.  
The proof of soundness is rather straighforward, thus we will focus on completeness. 
Following \cite{Wang15lori,Wang2016}, the strategy is to build, for any \KHiaxiom-consistent set of formulas, an \ults satisfying them. Note:

\begin{proposition}
  The following are theorems of \KHiaxiom:
  \begin{multicols}{2}
    \begin{itemize}
      \item[] \axm{SCOND}: \;\; $\vdash \A\lnot \psi \limp \khi(\psi, \varphi)$;
      \item[] \axm{COND}: \;\; $\vdash \khi(\bot, \varphi)$.
    \end{itemize}
  \end{multicols}
\end{proposition}

We proceed with the definition of the \emph{canonical model}.

\begin{definition}[Canonical model]\label{def:cm-ults-lkhi}
  Let $\smcs$ be the set of all maximally \KHiaxiom-consistent sets (MCS) of formulas in \KHilogic. For any $\Delta \in \smcs$, define 
  \begin{nscenter}
  $
    \begin{array}{r@{\;:=\;}l@{\qquad\qquad}r@{\;:=\;}l}
      \restkhi{\Delta}  & \csetc{\xi}{\Delta}{\xi \;\text{is of the form}\; \khi(\psi,\varphi)}, &
      \restkh{\Delta}   & \bigcup_{i \in \AGT} \restkhi{\Delta}.
    %    \restnkhi{\Delta} & \csetc{\xi}{\Delta}{\xi \;\text{is of the form}\; \lnot \khi(\psi,\varphi)}, &
   %   \restnkh{\Delta}  & \bigcup_{i \in \AGT} \restnkhi{\Delta}. \\
    \end{array}
    $
  \end{nscenter}

  Let $\Gamma$ be a set in $\smcs$. Define, for each agent $i \in \AGT$, the set of basic actions $\ACT^\Gamma_i := \csetsc{\tup{\psi,\varphi}}{\khi(\psi,\varphi) \in \Gamma}$, and $\ACT^\Gamma := \bigcup_{i \in \AGT} \ACT^\Gamma_i$. Notice that \axm{COND} implies that $\khi(\bot, \bot) \in \Gamma$ for every $i \in \AGT$; since there is at least one agent, this implies that $\ACT^\Gamma$ is non-empty, and thus it is an adequate set of actions. It is worth noticing that the set $\ACT^\Gamma$ fixes a new signature. However, since the operators of the language cannot see the names of the actions, we can define a mapping from $\ACT^\Gamma$ to any particular $\ACT$, to 
  preserve the original signature.
%  \fer{If the set $\ACT$ is already given, we could define a correspondence between $\ACT$ and $\ACT^\Gamma$, to guarantee that the model uses $\ACT$ instead.} 

  \noindent Then, the structure $\modults^\Gamma$, defined over $\ACT^\Gamma$, $\AGT$ and $\PROP$, is the tuple $\tup{\W^\Gamma, \R^\Gamma, \cset{\S^\Gamma_i}_{i \in \AGT}, \V^\Gamma}$ where
  \begin{compactitemize}
    \item $\W^\Gamma := \csetc{\Delta}{\smcs}{\restkh{\Delta} = \restkh{\Gamma}}$,

    \item $\R^{\Gamma}_{\tup{\psi,\varphi}} := \bigcup_{i \in \AGT} \R^{\Gamma} _{\tup{\psi,\varphi}^i}$, with $\R^{\Gamma}_{\tup{\psi,\varphi}^i} := \csetc{(\Delta_1, \Delta_2)}{\W^\Gamma \times \W^\Gamma}{\khi(\psi,\varphi) \in \Gamma, \psi \in \Delta_1, \varphi \in \Delta_2}$,

    \item $\S^\Gamma_i := \left\{ \cset{\tup{\psi,\varphi}} \mid \tup{\psi,\varphi} \in \ACT^\Gamma_i \right\}$,

    \item $V^\Gamma(\Delta) := \csetc{p}{\PROP}{p \in \Delta}$.
  \end{compactitemize}
\end{definition}

If $\Gamma \in \smcs$, then $\modults^\Gamma$ is a structure of the required type.

\begin{proposition}\label{pro:cm-ults-lkhi}
  The structure $\modults^\Gamma = \tup{\W^\Gamma, \R^\Gamma, \cset{\S^\Gamma_i}_{i \in \AGT}, \V^\Gamma}$ is an \ults.
  \begin{proof}
    It is enough to show that each $\S^\Gamma_i$ defines a partition over a non-empty subset of $\pow{\ACT^*}$. First, \axm{COND} implies $\khi(\bot, \bot) \in \Gamma$, so $\tup{\bot, \bot} \in \ACT^\Gamma_i$ and hence $\cset{\tup{\bot, \bot}} \in \S^\Gamma_i$; thus, $\bigcup_{\strategy \in \S_i} \strategy \neq \emptyset$. Then, $\S_i$ indeed defines a partition over $\bigcup_{\strategy \in \S_i} \strategy$: its elements are mutually disjoint (they are singletons with different elements), collective exhaustiveness is immediate and, finally, $\emptyset \notin \S^\Gamma_i$.
  \end{proof}
\end{proposition}

Let $\Gamma \in \smcs$, the following properties of $\modults^\Gamma$ are useful (proofs are similar to the ones in~\cite{Wang2016}).

\begin{proposition}\label{pro:cm-ults-lkhi-allsameKH}
  For any $\Delta_1, \Delta_2 \in \W^\Gamma$ we have $\restkh{\Delta_1} = \restkh{\Delta_2}$.
%  \begin{proof}
%    Straightforward from the definition of $\W^\Gamma$.
%  \end{proof}
\end{proposition}

\begin{proposition}\label{pro:cm-ults-lkhi-oneall}
  Take $\Delta \in \W^\Gamma$. If $\Delta$ has a $\R^\Gamma_{\tup{\psi,\varphi}}$-successor, then every $\Delta' \in \W^\Gamma$ with $\varphi \in \Delta'$ can be $\R^\Gamma_{\tup{\psi,\varphi}}$-reached from $\Delta$.
\end{proposition}

\begin{proposition}\label{pro:cm-ults-lkhi-allall}
  Let $\varphi$ be an \KHilogic-formula. If $\varphi \in \Delta$ for every $\Delta \in \W^\Gamma$, then $\A\varphi \in \Delta$ for every $\Delta \in \W^\Gamma$.
\end{proposition}

\begin{proposition}\label{pro:cm-ults-lkhi-succpre}
  Take $\psi, \psi', \varphi'$ in \KHilogic. Suppose that every $\Delta \in \W^\Gamma$ with $\psi \in \Delta$ has a $\R^{\Gamma}_{\tup{\psi',\varphi'}}$-successor. Then, $\A(\psi \limp \psi') \in \Delta$ for all $\Delta \in \W^\Gamma$.
\end{proposition}

With these properties at hand, we can prove the truth lemma for $\modults^\Gamma$.

\begin{lemma}[Truth lemma for $\modults^\Gamma$]\label{tlm:cm-ults-lkhi}
  Given $\Gamma \in \smcs$, take $\modults^\Gamma = \tup{\W^\Gamma, \R^\Gamma, \cset{\S^\Gamma_i}_{i \in \AGT}, \V^\Gamma}$. Then, for every $\Theta \in \W^\Gamma$ and every $\varphi \in \KHilogic$,   $\modults^\Gamma, \Theta \models \varphi$ if and only if $\varphi \in \Theta$. 
  % \begin{ctabular}{l@{\qquad\text{if and only if}\qquad}l}
  %   $\modults^\Gamma, \Theta \models \varphi$ & $\varphi \in \Theta$.
  % \end{ctabular}
  \begin{proof}
    The proof is by induction on $\varphi$, with the atom and Boolean cases as usual. For the rest:
    \begin{itemizenl}
      \item \textbf{Case $\boldsymbol{\khi(\psi,\varphi)}$.} {\prooflr} Suppose $\modults^\Gamma, \Theta \models \khi(\psi,\varphi)$, then consider two cases.
      \begin{itemize}
        \item $\bs{\truthset{\modults^\Gamma}{\psi} = \emptyset}$. Then, for each $\Delta \in \W^\Gamma$ we have $\Delta \not\in \truthset{\modults^\Gamma}{\psi}$, so $\psi \not\in \Delta$ (by IH) and thus $\lnot\psi \in \Delta$ (by maximal consistency). Hence, by \Cref{pro:cm-ults-lkhi-allall}, $\A\lnot\psi \in \Delta$ for every $\Delta \in \W^\Gamma$. In particular, $\A\lnot\psi \in \Theta$ and thus, by \axm{SCOND} and \axm{MP}, $\khi(\psi, \varphi) \in \Theta$.

        \item $\bs{\truthset{\modults^\Gamma}{\psi} \neq \emptyset}$. %From $\modults^\Gamma, \Theta \models \khi(\psi,\varphi)$ 
        By hypothesis, there is $\cset{\tup{\psi',\varphi'}} \in \S^\Gamma_i$ with \begin{inline-cond-kh} \item $\truthset{\modults^\Gamma}{\psi} \subseteq \stexec(\cset{\tup{\psi',\varphi'}})$ and \item $\R^\Gamma_{\cset{\tup{\psi',\varphi'}}}(\truthset{\modults^\Gamma}{\psi}) \subseteq \truthset{\modults^\Gamma}{\varphi}$\end{inline-cond-kh}. In other words, there is $\tup{\psi',\varphi'} \in \ACT^\Gamma_a$ such that
        \begin{cond-kh}
          \item\label{tlm:cm-esmiv-stexec-lkhi-itm:i} for all $\Delta \in \W^\Gamma$, if $\Delta \in \truthset{\modults^\Gamma}{\psi}$ then $\Delta \in \stexec(\cset{\tup{\psi',\varphi'}})$, so $\Delta \in \stexec(\tup{\psi',\varphi'})$ and therefore $\Delta$ has a $\R^\Gamma_{\tup{\psi',\varphi'}}$-successor.

          \item\label{tlm:cm-esmiv-stexec-lkhi-itm:ii} for all $\Delta' \in \W^\Gamma$, if there is $\Delta \in \truthset{\modults^\Gamma}{\psi}$ such that $(\Delta, \Delta') \in \R^\Gamma_{\tup{\psi',\varphi'}}$, then $\Delta' \in \truthset{\modults^\Gamma}{\varphi}$.
        \end{cond-kh}
        This case requires three pieces.
        \begin{enumerate}
          \item Take any $\Delta \in \W^\Gamma$ with $\psi \in \Delta$. Then, by IH, $\Delta \in \truthset{\modults^\Gamma}{\psi}$ and thus, by \Cref{tlm:cm-esmiv-stexec-lkhi-itm:i}, $\Delta$ has a $\R^\Gamma_{\tup{\psi',\varphi'}}$-successor. Thus, every $\Delta \in \W^\Gamma$ with $\psi \in \Delta$ has such successor; then (\Cref{pro:cm-ults-lkhi-succpre}), it follows that $\A(\psi \limp \psi') \in \Delta$ for every $\Delta \in \W^\Gamma$. In particular, $\A(\psi \limp \psi') \in \Theta$.

          \item From $\tup{\psi',\varphi'} \in \ACT^\Gamma_i$ it follows that $\khi(\psi',\varphi') \in \Gamma$. But $\Theta \in \W^\Gamma$, so $\restkh{\Theta} = \restkh{\Gamma}$ (by definition of $\W^\Gamma$). Hence, $\khi(\psi',\varphi') \in \Theta$.

          \item Since $\truthset{\modults^\Gamma}{\psi} \neq \emptyset$, there is $\Delta \in \truthset{\modults^\Gamma}{\psi}$. By \Cref{tlm:cm-esmiv-stexec-lkhi-itm:i}, $\Delta$ should have at least one $\R^\Gamma_{\tup{\psi',\varphi'}}$-successor. Then, by \Cref{pro:cm-ults-lkhi-oneall}, every $\Delta' \in \W^\Gamma$ satisfying $\varphi' \in \Delta'$ can be $\R^\Gamma_{\tup{\psi',\varphi'}}$-reached from $\Delta$; in other words, every $\Delta' \in \W^\Gamma$ satisfying $\varphi' \in \Delta'$ is in $\R^\Gamma_{\tup{\psi',\varphi'}}(\Delta)$. But $\Delta \in \truthset{\modults^\Gamma}{\psi}$, so every $\Delta' \in \W^\Gamma$ satisfying $\varphi' \in \Delta'$ is in $\R^\Gamma_{\tup{\psi',\varphi'}}(\truthset{\modults^\Gamma}{\psi})$. Then, by \Cref{tlm:cm-esmiv-stexec-lkhi-itm:ii}, every $\Delta' \in \W^\Gamma$ satisfying $\varphi' \in \Delta'$ is in $\truthset{\modults^\Gamma}{\varphi}$. By IH on the latter part, every $\Delta' \in \W^\Gamma$ satisfying $\varphi' \in \Delta'$ is such that $\varphi \in \Delta'$. Thus, $\varphi' \limp \varphi \in \Delta'$ for every $\Delta' \in \W^\Gamma$, and hence (\Cref{pro:cm-ults-lkhi-allall}) $\A(\varphi' \limp \varphi) \in \Delta'$ for every $\Delta' \in \W^\Gamma$. In particular, $\A(\varphi' \limp \varphi) \in \Theta$.
        \end{enumerate}
        Thus, $\cset{\A(\psi \limp \psi'), \khi(\psi',\varphi'), \A(\varphi' \limp \varphi)} \subset \Theta$. Therefore, by \axm{KhA} and \axm{MP}, $\khi(\psi, \varphi) \in \Theta$.
      \end{itemize}

      {\proofrl} Suppose $\khi(\psi, \varphi) \in \Theta$. Thus (\Cref{pro:cm-ults-lkhi-allsameKH}), $\khi(\psi, \varphi) \in \Gamma$, so $\tup{\psi, \varphi} \in \ACT^\Gamma_i$ and therefore $\cset{\tup{\psi, \varphi}} \in \S^\Gamma_i$. The rest of the proof is split in two cases.
      \begin{itemize}
        \item Suppose there is no $\Delta_\psi \in \W^\Gamma$ with $\psi \in \Delta$. Then, by IH, there is no $\Delta_\psi \in \W^\Gamma$ with $\Delta_\psi \in \truthset{\modults^\Gamma}{\psi}$, that is, $\truthset{\modults^\Gamma}{\lnot\psi} = \D{\W^\Gamma}$. Since $\modults^\Gamma$ is in \ults (\Cref{pro:cm-ults-lkhi}), the latter yields $(\modults^\Gamma, \Delta) \models \khi(\psi, \chi)$ for any $i \in \AGT$, $\chi \in \KHilogic$ and $\Delta \in \W^\Gamma$ (cf. \Cref{pro:ults:universal}); hence, $(\modults^\Gamma, \Theta) \models \khi(\psi, \varphi)$.

        \item Suppose there is $\Delta_\psi \in \W^\Gamma$ with $\psi \in \Delta_\psi$. It will be shown that the strategy $\cset{\tup{\psi, \varphi}} \in \S^\Gamma_i$ satisfies the requirements.
        \begin{cond-kh}
          \item Take any $\Delta \in \truthset{\modults^\Gamma}{\psi}$. By IH, $\psi \in \Delta$. Moreover, from $\khi(\psi, \varphi) \in \Theta$ and \Cref{pro:cm-ults-lkhi-allsameKH} it follows that $\khi(\psi, \varphi) \in \Delta$. Then, from $\R^\Gamma_{\tup{\psi,\varphi}^i}$'s definition, every $\Delta' \in \W^\Gamma$ with $\varphi \in \Delta'$ is such that $(\Delta, \Delta') \in \R^\Gamma_{\tup{\psi,\varphi}^i}$, and therefore such that $(\Delta, \Delta') \in \R^\Gamma_{\tup{\psi,\varphi}}$. Now note how, since there is $\Delta_\psi \in \W^\Gamma$ with $\psi \in \Delta_\psi$, there should be $\Delta_\varphi \in \W^\Gamma$ with $\varphi \in \Delta_\varphi$ (the proof uses \axm{KhE} and \axm{TA}).
          %\footnote{Suppose otherwise, i.e., suppose there is no $\Delta'' \in \W^\Gamma$ with $\varphi \in \Delta''$. Then, $\lnot \varphi \in \Delta''$ for every $\Delta'' \in \W^\Gamma$, and hence (\Cref{pro:cm-ults-lkhi-allall}) $\A\lnot\varphi \in \Delta''$ for every $\Delta'' \in \W^\Gamma$. In particular, $\A\lnot\varphi \in \Delta_\psi$. Moreover, from $\khi(\psi, \varphi) \in \Theta$ and \Cref{pro:cm-ults-lkhi-allsameKH} it follows that $\khi(\psi, \varphi) \in \Delta_\psi$. Then (\axm{KhE} and \axm{MP}), $\A\lnot\psi \in \Delta_\psi$ and thus (axiom \axm{TA}) $\lnot\psi \in \Delta_\psi$. Hence, $\cset{\psi, \lnot\psi} \subset \Delta_\psi$, contradicting $\Delta_\psi$'s consistency.} 
          This implies that $(\Delta, \Delta_\varphi) \in \R^\Gamma_{\tup{\psi,\varphi}}$ and thus, since $\tup{\psi,\varphi}$ is a basic action, $\Delta \in \stexec(\tup{\psi,\varphi})$ so $\Delta \in \stexec(\cset{\tup{\psi,\varphi}})$. Since $\Delta$ is an arbitrary state in $\truthset{\modults^\Gamma}{\psi}$, the required $\truthset{\modults^\Gamma}{\psi} \subseteq \stexec(\cset{\tup{\psi,\varphi}})$ follows.

          \item Take any $\Delta' \in \R^\Gamma_{\cset{\tup{\psi,\varphi}}}(\truthset{\modults^\Gamma}{\psi})$. Then, there is $\Delta \in \truthset{\modults^\Gamma}{\psi}$ such that $(\Delta, \Delta') \in \R^\Gamma_{\tup{\psi,\varphi}}$. By definition of $\R^\Gamma$, it follows that $\varphi \in \Delta'$ so, by IH, $\Delta' \in \truthset{\modults^\Gamma}{\varphi}$. Since $\Delta'$ is an arbitrary state in $\R^\Gamma_{\cset{\tup{\psi,\varphi}}}(\truthset{\modults^\Gamma}{\psi})$, the required $\R^\Gamma_{\cset{\tup{\psi,\varphi}}}(\truthset{\modults^\Gamma}{\psi}) \subseteq \truthset{\modults^\Gamma}{\varphi}$ follows.
        \end{cond-kh}
      \end{itemize}
    \end{itemizenl}
  \end{proof}
\end{lemma}

%% file: complexity.tex
%!TEX root = tark21.tex
Here we investigate the computational complexity of the satisfiability
problem of $\KHilogic$ under the $\ults$-based semantics.  We will
establish membership in \NP by showing a polynomial size model property. 

Given a formula, we will show that it is possible to select just a
piece of the canonical model which is relevant for its evaluation. The
selected model will preserve satisfiability, and moreover, its size
will be polymonial w.r.t.\ the size of the input formula.

\begin{definition}[Selection function]
  \label{def:selection-function}
  Let
  $\cmodel=\tup{\W^\Gamma,\R^\Gamma,\cset{\S_i^\Gamma}_{i\in\AGT},\V^\Gamma}$
  be a canonical model for an MCS $\Gamma$ (see
  \Cref{def:cm-ults-lkhi}); take an MCS $w\in\W^\Gamma$ and a formula
  $\varphi\in\KHilogic$. Define
  $\ACT_\varphi := \cset{\tup{\theta_1,\theta_2} \in\ACT^\Gamma \mid
    \kh_i(\theta_1,\theta_2) \text{ is a subformula of } \varphi}$.  A
  \emph{canonical selection function} $\sel^{\varphi}_{w}$ is a
  function that takes $\cmodel$, $w$ and $\varphi$ as input, returns a
  set $\W'\subseteq \W^\Gamma$, and is s.t.:
    
  \begin{enumerate} \itemsep 0cm
  \item \label{def:s-f1} $\sel^{\varphi}_{w}(p)=\{w\}$; \quad \quad
    $\sel^{\varphi}_{w}(\neg\varphi_1)=\sel^\varphi_w(\varphi_1)$;
    \quad \quad
    $\sel^{\varphi}_{w}(\varphi_1\vee\varphi_2)=
    \sel^{\varphi}_{w}(\varphi_1)\cup\sel^{\varphi}_{w}(\varphi_2)$;
        %\item \label{def:s-f2} $\sel^{\varphi}_{w}(\neg\varphi_1)=\sel^\varphi_w(\varphi_1)$
      %  \item \label{def:s-f3} $\sel^{\varphi}_{w}(\varphi_1\wedge\varphi_2)=
         % \sel^{\varphi}_{w}(\varphi_1)\cup\sel^{\varphi}_{w}(\varphi_2)$;
    
  \item \label{def:s-f4} If
    $\truthset{\cmodel}{\kh_i(\varphi_1,\varphi_2)}\neq\emptyset
    \text{ and } \truthset{\cmodel}{\varphi_1}=\emptyset$:
    $\sel^{\varphi}_{w}(\kh_i(\varphi_1,\varphi_2)) = \{w\}$;
    
  \item \label{def:s-f5} If
    $\truthset{\cmodel}{\kh_i(\varphi_1,\varphi_2)}\neq\emptyset
    \text{ and } \truthset{\cmodel}{\varphi_1}\neq\emptyset$:
        
    $\sel^{\varphi}_{w}(\kh_i(\varphi_1,\varphi_2)) = \{w_1,w_2\} \cup
    \sel^{\varphi}_{w_1}(\varphi_1) \cup
    \sel^{\varphi}_{w_2}(\varphi_2)$,
    \noindent where $w_1$, $w_2$ are
    s.t. $(w_1,w_2)\in\R^\Gamma_{\tup{\varphi_1,\varphi_2}}$;
    
  \item \label{def:s-f6} If
    $\truthset{\cmodel}{\kh_i(\varphi_1,\varphi_2)}=\emptyset$ (note
    that $\truthset{\cmodel}{\varphi_1}\neq\emptyset$):
        
    For all set of plans $\strategy$, either
    $\truthset{\cmodel}{\varphi_1} \not\subseteq \stexec(\strategy)$
    or
    $\R_\strategy^\Gamma(\truthset{\cmodel}{\varphi_1}) \not\subseteq
    \truthset{\cmodel}{\varphi_2}$. For each $a \in \ACT_\varphi$:

    \begin{enumerate}
    \item if
      $\truthset{\cmodel}{\varphi_1} \not\subseteq \stexec(\cset{a})$:
      we add $\{w_1\} \cup \sel^{\varphi}_{w_1}(\varphi_1)$ to
      $\sel^{\varphi}_{w}(\kh_i(\varphi_1,\varphi_2))$, where
      $w_1 \in\truthset{\cmodel}{\varphi_1}$ and
      $w_1 \notin \stexec(\cset{a})$;
        
    \item if
      $\R_\strategy^\Gamma(\truthset{\cmodel}{\varphi_1})
      \not\subseteq \truthset{\cmodel}{\varphi_2}$ we add
      $\{w_1,w_2\} \cup \sel^{\varphi}_{w_1}(\varphi_1) \cup
      \sel^{\varphi}_{w_2}(\varphi_2)$ to
      $\sel^{\varphi}_{w}(\kh_i(\varphi_1,\varphi_2))$, where
      $w_1\in\truthset{\cmodel}{\varphi_1}$,
      $w_2 \in \R^\Gamma_a(w_1)$ and
      $w_2 \notin \truthset{\cmodel}{\varphi_2}$.
    \end{enumerate}
  \end{enumerate}
\end{definition}
    
We can now select a small model which preserves the
satisfiability of a given formula.

\begin{definition}[Selected model]
  Let $\cmodel$ be the canonical model for an MCS $\Gamma$, $w$ a
  state in $\cmodel$, and $\varphi$ an $\KHilogic$-formula. Let
  $\sel^\varphi_w$ be a selection function, we define the
  \emph{model selected by} $\sel^\varphi_w$ as
  $\model^\varphi_w=\tup{\W^\varphi_w,\R^\varphi_w,\cset{(\S^\varphi_w)_i}_{i\in\AGT},\V^\varphi_w}$,
  where

    \begin{compactitemize}
      \item $\W^\varphi_w := \sel^\varphi_w(\varphi)$; 
      \item $(\R^\varphi_w)_{\tup{\theta_1,\theta_2}} := \R^\Gamma_{\tup{\theta_1,\theta_2}}\cap(\W^\varphi_w)^2$ for each $\tup{\theta_1,\theta_2} \in \ACT_\varphi$;
      \item $(\S^\varphi_w)_i := \cset{\cset{a} \ |\ a \in \ACT_\varphi} \cup \cset{\cset{\tup{\bot,\top}}}$, for $i\in\AGT$
        (and $(\R^\varphi_w)_{\tup{\bot,\top}} := \emptyset$);
        \item $\V^\varphi_w$ is the restriction of $\V^\Gamma$ to $\W^\varphi_w$.
     % \item $\V^\varphi_w$ is the restriction of $\V^\Gamma$ to $\W^\varphi_w$.
    \end{compactitemize}
    \end{definition}
    
    Note that, although $\ACT_\varphi$ can be an empty set, each collection of sets of plans $(\S^\varphi_w)_i$ is not. Therefore, $\model^\varphi_w$ is an \esm.
    
    \begin{proposition}\label{prop:selection-preserves-sat}
    Let $\cmodel$ be a canonical model, $w$ a state in $\cmodel$ and $\varphi$ an $\KHilogic$-formula. Let $\model^\varphi_w$ be the selected model by a selection function $\sel^\varphi_w$.
    $\cmodel,w\models\varphi$ implies that 
    for all $\psi$ subformula of $\varphi$, and for all $v\in\W^\varphi_w$ we have that 
    $\cmodel,v\models\psi$ iff $\model^\varphi_w,v\models\psi$.
    % Let $\sel^\varphi_w$ be a selection function for $\cmodel$, $w$ and $\varphi$. If $\cmodel,w\models \varphi$, iff the generated model by $\sel^{\varphi}_{w}(\varphi)$ satisfies $\varphi$. 
    Moreover, $\model^\varphi_w$ is polynomial on the size of $\varphi$.
    \end{proposition}

    \input{proofcomplexity-short}

  In order to prove that the satisfiability problem of $\KHilogic$
  is in \NP, it remains to show that the model checking problem is in \Poly.

\begin{proposition}\label{prop:modcheck}
  The model checking problem for $\KHilogic$ is in \Poly.
\end{proposition}

\begin{proof}
  Given a pointed \ults $\model,w$ and a formula $\varphi$, we define
  a bottom-up labeling algorithm running in polymonial time which checks whether 
  $\model,w\models\varphi$. We follow
  the same ideas as for the basic modal logic {\sf K} (see
  e.g.,~\cite{blackburn06}). Below we introduce the case for formulas
  of the shape $\kh_i(\psi,\varphi)$, over an \ults
  $\model=\tup{\W,\R,\S,\V}$:

  \begin{algorithmic}
  \small
    \State{{\bf Procedure} ModelChecking($(\model,w)$, $\kh_i(\psi,\varphi))$} 
    \State $lab(\kh_i(\psi,\varphi))\gets \emptyset$;
    \ForAll{$\strategy\in\S_i$} 
      \State{$kh \gets True$}; 
      \ForAll{$\sigma\in\strategy$}
        \ForAll{$v\in lab(\psi)$}
          \State{$kh\gets (kh \ \& \ v\in\stexec(\sigma) \ \& \ \R_\sigma(v)\subseteq lab(\varphi))$}; 
        \EndFor
      \EndFor

      \If{$kh$}
        \State $lab(\kh_i(\psi,\varphi))\gets \W$;
      \EndIf
    \EndFor 
  \end{algorithmic}

  As $\S_i$ and each $\strategy\in\S_i$ are not empy, the first two
  {\bf for} loops are necessarily executed.  If
  $lab(\psi) =\emptyset$, then the formula $\kh_i(\psi,\varphi)$ is
  trivally true. Otherwise, $kh$ will remain true only if the
  appropriate conditions for the satisfiability of
  $\kh_i(\psi,\varphi))$ hold. If no $\strategy$ succeeds, then the
  initialization of $lab(\kh_i(\psi,\varphi))$ as $\emptyset$ will not
  be overwritten, as it should be.  Both $v\in\stexec(\sigma)$ and
  $\R_\sigma$ can be verified in polynomial time. Hence, the model
  checking problem is in \Poly.
\end{proof}

The intended result for satisfiability now follows.

    \begin{theorem}\label{th:khcomplexity}
    The satisfiability problem for $\KHilogic$ over $\ultss$ is \NP-complete.
    \begin{proof}  
      Hardness follows from \NP-completeness of propositional logic (a fragment of $\KHilogic$).
      By \Cref{prop:selection-preserves-sat}, each satisfiable formula 
      $\varphi$ has a model of polynomial size on $\varphi$. Thus, we can guess 
      a polymonial model $\model,w$, and verify $\model,w\models\varphi$ (which can 
      be done in polyonomial time, due to~\Cref{prop:modcheck}). %similarly as for modal logic {\sf K}~\cite{blackburn06}) (see~\Cref{sec:appendix} for details). 
      Thus, the satisfiability problem is in the class \NP.
    \end{proof}
    \end{theorem}

%%% Local Variables:
%%% mode: latex
%%% TeX-master: "tark21"
%%% End:

%% file: proofcomplexity-short.tex
\begin{proof}
The proof proceeds by induction in the size of the formula. Boolean cases are simple, so we will 
proceed with the case in which $\psi=\kh_i(\psi_1,\psi_2)$. 
%
% \begin{compactitemize}
% \item \textbf{Case $\bs{\psi=\kh_i(\psi_1,\psi_2)}$:}
% \subitem $\cmodel,v\models \kh_i(\psi_1,\psi_2)$ iff
% there is a $\strategy \in \S^\Gamma_i$ s.t.
% \subitem \quad $\truthset{\cmodel}{\psi_1} \subseteq \stexec^{\cmodel}(\strategy)$ and
% $\R^\Gamma_\strategy(\truthset{\cmodel}{\psi_1}) \subseteq \truthset{\cmodel}{\psi_2}$
%
% \subitem $\model^\varphi_w,v\models \kh_i(\psi_1,\psi_2)$ iff
% there is a $\strategy' \in (\S^\varphi_w)_i$ s.t.
% \subitem \quad $\truthset{\model^\varphi_w}{\psi_1} \subseteq \stexec^{\model^\varphi_w}(\strategy')$ and
% $(\R^\varphi_w)_{\strategy'}(\truthset{\model^\varphi_w}{\psi_1}) \subseteq \truthset{\model^\varphi_w}{\psi_2}$
%% \end{compactitemize}
%

Suppose that $\cmodel,v\models \kh_i(\psi_1,\psi_2)$. Then, we have two cases:

%       \item \textbf{Case $\boldsymbol{\khi(\psi,\varphi)}$.} {\prooflr} Consider two cases.
%       \begin{itemize}
%         \item $\bs{\truthset{\modults^\Gamma}{\psi} = \emptyset}$

\begin{itemize}
\item $\bs{\truthset{\cmodel}{\psi_1}\neq\emptyset}$:
by $\cmodel,v\models \kh_i(\psi_1,\psi_2)$, 
there exists a $\strategy \in \S^\Gamma_i$ s.t.
$\truthset{\cmodel}{\psi_1} \subseteq \stexec^{\cmodel}(\strategy)$ and
$\R^\Gamma_\strategy(\truthset{\cmodel}{\psi_1}) \subseteq \truthset{\cmodel}{\psi_2}$. 
By Truth Lemma, $\kh_i(\psi_1,\psi_2) \in v$, then $\kh_i(\psi_1,\psi_2) \in \Gamma$ and $\tup{\psi_1,\psi_2} \in\ACT_\Gamma$.
By the definition of $\R^\Gamma_{\tup{\psi_1,\psi_2}}$,
we have that for all $w \in \truthset{\cmodel}{\psi_1}$, it holds that
$\R^\Gamma_{\tup{\psi_1,\psi_2}}(w)\neq\emptyset$
and $\R^\Gamma_{\tup{\psi_1,\psi_2}}(w) \subseteq \truthset{\cmodel}{\psi_2}$.
Thus, 
$\truthset{\cmodel}{\psi_1} \subseteq \stexec^{\cmodel}(\cset{\tup{\psi_1,\psi_2}})$ and $\R^\Gamma_{\tup{\psi_1,\psi_2}}(\truthset{\cmodel}{\psi_1}) \subseteq \truthset{\cmodel}{\psi_2}$. 
Since $\truthset{\cmodel}{\psi_1}\neq\emptyset$, 
there exist $w_1,w_2\in\W^\Gamma$ s.t.
$(w_1,w_2) \in\R^\Gamma_{\tup{\psi_1,\psi_2}}$.

Notice that by definition of $\model^\varphi_w$, we have that
$\cset{\tup{\psi_1,\psi_2}} \in (\S^\varphi_w)_i$ and  that $(\R^\varphi_w)_{\tup{\psi_1,\psi_2}}$ is defined. Also, by the definition of $\sel^{\varphi}_{w}$, \Cref{def:s-f5}, there exist $w'_1,w'_2\in\W^\varphi_w$ s.t. $(w'_1,w'_2)\in(\R^\varphi_w)_{\tup{\psi_1,\psi_2}}$.
Let $v_1\in \truthset{\model^\varphi_w}{\psi_1}$
$\subseteq \truthset{\model^\Gamma}{\psi_1}$ (the inclusion holds by IH). Then, we have $v_1 \in \stexec^{\cmodel}(\cset{\tup{\psi_1,\psi_2}})$ and $\R^\Gamma_{\tup{\psi_1,\psi_2}}(v_1) \subseteq \truthset{\cmodel}{\psi_2}$.
Since for all $v_2 \in \R^\Gamma_{\tup{\psi_1,\psi_2}}(v_1)$, we have $v_2\in\truthset{\cmodel}{\psi_2}$, (in particular $v_2=w'_2$), then $w'_2 \in(\R^\varphi_w)_{\tup{\psi_1,\psi_2}}(v_1)$.
Thus, $v_1 \in \stexec^{\model^\varphi_w}(\cset{\tup{\psi_1,\psi_2}})$. 

Aiming for a contradiction, suppose now that $(\R^\varphi_w)_{\tup{\psi_1,\psi_2}}(v_1) = \R^\Gamma_{\tup{\psi_1,\psi_2}}(v_1)\cap \W^\varphi_w \not\subseteq \truthset{\model^\varphi_w}{\psi_2}$; and 
let $v_2\in(\R^\varphi_w)_{\tup{\psi_1,\psi_2}}(v_1)$ s.t. $v_2 \not\in \truthset{\model^\varphi_w}{\psi_2}$. Then we have that $(\R^\varphi_w)_{\tup{\psi_1,\psi_2}}(v_1) \subseteq \R^\Gamma_{\tup{\psi_1,\psi_2}}(v_1)$, but also by IH $v_2 \not\in \truthset{\cmodel}{\psi_2}$. Thus, $\cmodel,v\not\models \kh_i(\psi_1,\psi_2)$, which is a contradiction. 
Then, it must be the case that $(\R^\varphi_w)_{\cset{\tup{\psi_1,\psi_2}}}(v_1) \subseteq \truthset{\model^\varphi_w}{\psi_2}$. 
Since we showed that $\truthset{\model^\varphi_w}{\psi_1} \subseteq \stexec^{\model^\varphi_w}(\cset{\tup{\psi_1,\psi_2}})$ and
$(\R^\varphi_w)_{\cset{\tup{\psi_1,\psi_2}}}(\truthset{\model^\varphi_w}{\psi_1}) \subseteq \truthset{\model^\varphi_w}{\psi_2}$, we can conclude $\model^\varphi_w,v\models \kh_i(\psi_1,\psi_2)$.

\item $\bs{\truthset{\cmodel}{\psi_1}=\emptyset}$: this case is direct.

\end{itemize}

Suppose now that $\model^\varphi_w,v\models \kh_i(\psi_1,\psi_2)$:

\begin{itemize}
\item $\bs{\truthset{\model^\varphi_w}{\psi_1}\neq\emptyset}$:  first, notice that
by IH,  $\truthset{\cmodel}{\psi_1}\neq\emptyset$. Also, by
 $\model^\varphi_w,v\models \kh_i(\psi_1,\psi_2)$, we get 
 $\truthset{\model^\varphi_w}{\psi_1} \subseteq \stexec^{\model^\varphi_w}(\strategy')$
and $(\R^\varphi_w)_{\strategy'}(\truthset{\model^\varphi_w}{\psi_1}) \subseteq \truthset{\model^\varphi_w}{\psi_2}$, for some $\strategy' \in (\S^\varphi_w)_i$. 
Aiming for a contradiction, suppose $\cmodel,v\not\models \kh_i(\psi_1,\psi_2)$. 
This implies that for all $\strategy \in \S^\Gamma_i$, 
$\truthset{\cmodel}{\psi_1} \not\subseteq \stexec^{\cmodel}(\strategy)$ or
$\R^\Gamma_\strategy(\truthset{\cmodel}{\psi_1}) \not\subseteq \truthset{\cmodel}{\psi_2}$. 
Also, by definition of $\ACT_\varphi$ we have that for all $\strategy=\cset{a} \in (\S^\varphi_w)_i$, with $a \in \ACT_\varphi$, 
$\truthset{\cmodel}{\psi_1} \not\subseteq \stexec^{\cmodel}(\strategy)$ or
$\R^\Gamma_\strategy(\truthset{\cmodel}{\psi_1}) \not\subseteq \truthset{\cmodel}{\psi_2}$;
i.e., for all $a \in \ACT_\varphi$
$\truthset{\cmodel}{\psi_1} \not\subseteq \stexec^{\cmodel}(\cset{a})$ or $\R^\Gamma_{\cset{a}}(\truthset{\cmodel}{\psi_1}) \not\subseteq \truthset{\cmodel}{\psi_2}$.
Thus, there exists $w_1 \in\truthset{\cmodel}{\psi_1}$ s.t.
$w_1 \not\in \stexec^{\cmodel}(a)$ or there exists $w_2 \in\R^\Gamma_a(w_1)$ s.t. $w_2 \not\in \truthset{\cmodel}{\psi_2}$.
By definition of $\sel^{\varphi}_{w}$, \Cref{def:s-f6}, we add witnesses for each $a \in \ACT_\varphi$. So, let $\strategy' \in (\S^\varphi_w)_i$. If $\strategy'=\cset{\tup{\bot,\top}}$,
trivially we obtain 
$\emptyset \neq\truthset{\model^\varphi_w}{\psi_1}\not\subseteq \stexec^{\model^\varphi_w}(\strategy')=\emptyset$.
Then, take another $\strategy'=\cset{a}$ s.t. $a\in\ACT_\varphi$,
and $w'_1 \in \truthset{\model^\varphi_w}{\psi_1} \subseteq \truthset{\cmodel}{\psi_1}$. 
If $w'_1 \not\in \stexec^{\cmodel}(\cset{a})$, $\R^\Gamma_a(w'_1)=\emptyset$ and thus $(\R^\varphi_w)_a(w'_1)=\emptyset$ and therefore $w'_1 \not\in \stexec^{\model^\varphi_w}(\cset{a})$. 
On the other hand, if there exists $w_2 \in\R^\Gamma_a(w'_1)$ s.t. $w_2 \not\in \truthset{\cmodel}{\psi_2}$, then by $\sel^{\varphi}_{w}$ and IH,
there exists $w'_2 \in\W^\varphi_w$ s.t.
$w'_2 \in\R^\Gamma_a(w'_1)$ and $w'_2\not\in\truthset{\model^\varphi_w}{\psi_2}$, and 
consequently, there exists $w'_2 \in(\R^{\varphi}_{w})_a(w'_1)$ s.t. $w'_2 \not\in \truthset{\model^{\varphi}_{w}}{\psi_2}$. In any case, it leads to $\model^{\varphi}_{w},v\not\models \kh_i(\psi_1,\psi_2)$, a contradiction. Therefore, $\cmodel,v\models \kh_i(\psi_1,\psi_2)$.

\item $\bs{\truthset{\model^\varphi_w}{\psi_1}=\emptyset}$: similar to the previous case.

\end{itemize}

Thus, we proved the case $\cmodel,v\models \kh_i(\psi_1,\psi_2)$ iff
$\model^\varphi_w,v\models \kh_i(\psi_1,\psi_2)$. Therefore, we get that 
for all $\psi$ subformula of $\varphi$ and $v\in\W^\varphi_w$, $\cmodel,v\models\psi$ iff $\model^\varphi_w,v\models\psi$.
  Notice that the selection funcion adds worlds from $\cmodel$, only for each $\khi$-formula that 
  appears as a subformula of $\varphi$. Clearly, there is a polynomial number of such
  subformulas.
  Moreover, the number of worlds added at each time is also polynomial 
  in the size of $\varphi$. Hence, $\W^\varphi_w$ is of polynomial size. Since $(\S^\varphi_w)_i$ is also polynomial, we have that the size of $\model^\varphi_w$ is polynomial in the size of $\varphi$.

\end{proof}

%% file: final.tex
In this article, we introduce a new semantics for the \emph{knowing
  how} modality from~\cite{Wang15lori,Wang16,Wang2016}, over multiple
agents. It is defined in terms of \emph{uncertainty-based labeled
  transition systems (\ults)}. The novelty in our proposal is that
\ults{s} are equipped with an indistinguishability relation among
plans.  In this way, the epistemic notion of uncertainty of an agent
--which in turn defines her epistemic state-- is reintroduced, bringing the notion of \emph{knowing how} closer to the notion of \emph{knowing that} from classical epistemic
logics.  We believe that the semantics based on \esm can represent properly the
situation of a shared, objective description of the affordances of a
given situation, together with the different, subjective and personal
abilities of a group of agents; this seems difficult to achieve using
a semantics based on LTSs alone.

We show that the logic of~\cite{Wang15lori,Wang16,Wang2016} can be obtained
by imposing particular conditions over \ults{s}; thus, the new semantics is more general.
In particular, it provides counter-examples to \axm{EMP} and \axm{COMP}, which  directly link $\kh$ to
properties of the universal modality.\footnote{The rest of the axioms and rules in \KHaxiom  (those shown in block $\axset$) merely state properties of the universal modality and the fact that $\kh$ is global.}
Indeed, consider \axm{EMP}: even though $\A(\psi\ra\varphi)$ objectively holds in the underlying
LTS of an \esm, it could be argued that an agent might not have actions or plans at her disposal to turn those facts into knowledge, resulting in $\kh(\psi,\varphi)$ failing on the
model. Moreover, we have introduced a sound and strongly complete axiom system for the new semantics over \esm{s}. 
Finally, we showed that the satisfiability problem for our multi-agent knowing
how logic over the \ults-based semantics is \NP-complete, via a selection argument (and
model checking is polynomial).

% Finally, we exemplify how the new semantics enables the definition of
% dynamic operations that update the knowledge of an agent. 
% It is natural to define operators 
% updating the knowledge of an agent by reducing her \emph{uncertainty}. In the 
% LTS-based semantics, one can model these updates by removing  
% actions that violate the strong executability of a plan; but it is difficult to control such removal, and ensure that it does not alter the strong executability statu	s of other, unrelated, plans. Similarly, one could add new strongly executable plans, but again this can lead to side effects.
% To exemplify the dynamic possibilities of the \esm based semantics, 
% we defined
% a modality which relies on updates in the
% indistinguishability relation between plans.  
% We also consider some variants of the operator, ranging
% from an (in a multi-agent setup, public) announcement to
% distinguish two specific plans, arbitrary announcements, and a goal-directed learning operator.
% To the best of our knowledge, this article investigates, for the first time, 
% dynamic modalities in the context of a \emph{knowing how} logic. 

\paragraph{Future work.}
There are several interesting lines of research to explore in the
future.  First, our framework easily accommodates other notions of
executability. For instance, one could require only some of the plans in
a set $\strategy$ to be strongly executable, weaken the condition
of \emph{strong} executability, etc. We can also explore the 
effects of imposing different restrictions
on the construction of the indistinguishability relation between plans. 
  It would be interesting to investigate which logics we
obtain in these cases, and their relations with the \lts
semantics. 

Second, to our knowledge, the exact complexity of
the satisfiability problem for knowing how over \lts{s} is open.
It would be interesting to see whether an adaptation of our selection
procedure works over \lts{s}. 

Third, the \esm semantics, in the multi-agent
setting, leads to natural definitions of concepts such as
\emph{global}, \emph{distributed} and \emph{common knowing how}, which
should be investigated in detail.
% Second, in this article, we barely scratch
% the surface of the dynamic possibilities offered by the new semantics.
% We would like to investigate properties like axiomatization,
% expressivity and complexity of the dynamic operators presented.
%% RAUL: commented for long version
% Regarding axiomatizations, one of the main obstacles towards getting a
% complete axiomatization for these dynamic modalities is that the rule
% of uniform substitution fails. This behaviour is not surprising: when
% dealing with modalities for model-changing operations, we usually lose
% uniform substitution (see, e.g.,~\cite{HollidayHI13,ArecesFH15}).
% Regarding expressivity, we would like to know the exact relationship
% between the modalities we proposed, as this issue turned to be
% particularly challenging.  We plan to approach this issue by defining
% Ehrenfeucht-Fra\"iss\'e games to capture their respective
% expressivity.

Finally, dynamic modalities capturing epistemic updates can be defined
via operations that modify the indistinguishability relation among
plans (as is done with other dynamic epistemic
operators, see, e.g.,~\cite{DELbook}).  This would allow to express
different forms of communication, such as \emph{public},
\emph{private} and \emph{semi-private} announcements concerning (sets of) plans.

%%% Local Variables:
%%% mode: latex
%%% TeX-master: "tark21"
%%% End: